\documentclass[11pt,leqno]{article}
\usepackage{amsmath,amssymb,amsfonts}
\newtheorem{theorem}{Theorem}[section]
\newtheorem{proposition}[theorem]{Proposition}

\newtheorem{definition}[theorem]{Definition}
\newtheorem{remark}[theorem]{Remark}
\newtheorem{notation}[theorem]{Notation}
\newtheorem{example}[theorem]{Example}

\newenvironment{proof}{\mbox{\bf Proof.}}{\mbox{$\dashv$}\bigskip}

\begin{document}

\begin{center}

{\Large\bf Computation Environments (2)}\\ { Persistently Evolutionary Semantics}\\
\vspace{.25in}
{\bf  Rasoul Ramezanian}\\
Department of Mathematical Sciences, \\
Sharif University of Technology,\\
P. O. Box 11365-9415, Tehran, Iran\\
 ramezanian@sharif.edu
\end{center}
\begin{abstract}
\noindent In the manuscript titled ``Computation environment (1)",
we introduced a notion called computation environment as an
interactive model for computation and complexity theory.  In this
model, Turing machines are not autonomous entities and   find
their meanings through  the interaction between a computist and a
universal processor, and thus due to evolution of the universal
processor, the meanings of   Turing machines could change. In this
manuscript, we discuss persistently evolutionary intensions. We
introduce a new semantics, called persistently evolutionary
semantics,   for predicate logic that the
 meaning of function  and predicate  symbols are not already predetermined,
 and
 predicate and function symbols find their meaning through the
 interaction of the subject with the language. In (classic) model theory,
the mathematician who studies a structure is assumed as a god who
lives out of the structure, and the study of the mathematician
does not effect the structure. The meaning of predicate and
function symbols are assumed to be independent of the
mathematician who does math. The persistently evolutionary
semantics could be regarded as a start of    ``Interactive Model
Theory" as a new paradigm in model theory (similar to the paradigm
of interactive computation). In interactive model theory, we
suppose that a mathematical structure should consist of two parts:
1) an intelligent agent (a subject), and 2) an environment
(language), and every things should find its meaning through the
interaction of these two parts.

We introduce   persistently evolutionary Kripke structure  for
propositional and predicate logic. Also, we propose a persistently
evolutionary Kripke semantics  for the notion of computation,
where the intension of a code of a Turing machine persistently
evolve. We show that in this Kripke model the subject can never
know $\mathrm{P=NP}$.

\end{abstract}
\begin{itemize}
\item[] \textbf{Keywords}: Interactive Model Theory, Free will,
Persistent Evolution.
\end{itemize}
\section{Introduction}

The human being  interacts with its surrounded environment, and
through this interaction, the environment finds its meaning for
him. He percepts its environment (that he calls it the real world)
through its sensors, and in order to understand it, he constructs
mental symbolic forms~\cite{kn:kasir}, and formal structures   in
his mind. After that, he reasons about its environment through
these mental constructions.

 Euclidian Geometry, Ptolemy's Almagset, Copernicus
revolution, Turing computation, and etc are some few samples of
theories constructed by the human being  through interaction with
the world.

After the  human being constructs a formal theory, he observes
  and means its environment through it as a window to the outer world. He proposes questions about its
   environment in his (constructed mental) formal theory, and  tries to answer
  them in the context of the same   theory. For
  example, in astronomy, after Ptolemy's Almagset, the human being
  tried to find an explanation for the irregular motions of the
  \emph{wandering stars} in the context of Ptolemy's Almagset,
  and in geometry  the human being attempted to prove \emph{parallel postulate} using
  Euclid's first four postulates.

Model theory, a branch of mathematical logic~\cite{kn:model},
studies mathematical structure  in order to determine  that given
a   theory (a set of formulas) $\Gamma$, which other formulas are
true in all structures which all formulas in $\Gamma$ are true in
them. In (classic) model theory, the mathematician   is regarded
as a god who lives out of a mathematical structure, and thinking
activities of the mathematician does not effect the structure. In
this way, the role of the human being in developing   formal
systems is ignored. In this paper, we introduce a new semantics
for  predicate logic, that the meaning of predicate and function
symbols are not already predetermined, and they find their meaning
through the interaction of a subject with the logical language. We
name the proposed semantics ``the persistently evolutionary
semantics".

The paper is organized as follows:

\noindent In Section~\ref{CID}, We discuss persistently
evolutionary intensions.  We scrutinize that whether it is
possible that the intension of a word persistently evolves whereas
the subject cannot be aware of it.

\noindent In section~\ref{SEM}, we propose persistently
evolutionary Kripke semantics to formalize the notion of
persistently evolutionary intensions.

\noindent In section~\ref{Logic}, and section~\ref{Logic2},  we
introduce persistently evolutionary semantics for propositional
and predicate logic.

\noindent In section~\ref{compu}, using persistently evolutionary
semantics for predicate logic, we formalize the argument of
section~7 of the manuscript~\cite{kn:comp1}.

\section{Persistently Evolutionary Intensions}\label{CID}

The human being, as an intelligent agent, uses   languages to
express and encode the intensions (concepts) that he constructs to
mean his environment.     Intension refers to a property that
specifies the set of all possible things that a \emph{word} (a
finite string) could describe, while extension refers to the set
of all actual things the word describes. Also an intensional
definition of a set of objects is to intend the set by a word, and
an extensional definition of a set of objects is by listing all
objects. Obviously, it is impossible to give an extensional
definition for an infinite set. For example, the human being
intends  an infinite subset of natural numbers by the word
``prime", and he can never list all prime numbers. As another
example, the human being, to define the set of all Turing machines
has no way except to use the intensional definition.

  In
 theory of computation, for every Turing machine $T$, $L(T)$ refers
to the set of all strings that the Turing machine $T$ halts for
them. We may say the Turing machine $T$ is an intensional
definition for the set $L(T)$, or in other words, the human being
intends the set $L(T)$ by the \emph{word} (finite string) $T$
(note that Turing machines can be coded in finite strings).

The human being   constructs concepts to mean its environment
through them. It is possible that both the human being and its
environment (the real world, the nature) evolve as the human being
checks that whether an object   is an extension of a word. This
evolution could be persistent such that the nature (or the human
being)  works well-defined. That is, if an output  $z$ is already
provided for an input $[x,y]$ ($x$ is a word, and $y$ is a thing
in order to be checked whether it is an extension of $x$)
 then whenever in future,  the same input
$[x,y]$ are chosen, the  output would be the same $z$. In other
words, the meaning of the word $x$ may change, but in a
conservative manner, that is, all the things  that the human being
already realized that whether they are extensions of $x$ or not,
their status  remains unchanged.

\begin{definition}\label{ped}
Let $w$ be a word (a finite string), and $\mathcal{O}$ a domain of
objects. We say the intension of the word $w$ for a subject is
persistently evolutionary (or its extension is order-sensitive)
  whenever in the course that
   the subject chooses an object $o\in\mathcal{O}$ to
  check that whether it is an extension of the word $w$ or not, then the intension of $w$   changes,
   but persistently,
i.e., if the agent (the subject)  has checked whether an object
$d$ is an extension of $w$ already, and the answer has been yes
(has been no), then whenever in future the agent checks again
whether the same object $d$ is an extension, the answer would be
the same yes (would be the same no). What remains unchanged is the
word $w$ (the syntax), but its meaning (the semantics) changes for
the subject. In this way, the set of all extensions of the word
$w$ is not predetermined and it depends on the order that the
agent chooses objects from the domain $\mathcal{O}$ to check
whether they are extensions of $w$ or not.

\end{definition}

\begin{itemize}
\item[\textbf{Q1)}] Is it possible that an intension of a word
would be a persistently evolutionary one?
\end{itemize}The answer is Yes.
To answer the above question, we should first clarify  what the
meaning (the intension) of a word is?  For a
  subject (the human being) the meaning of a word is given
by how the  subject interacts via the word with the environment
(the language). This interaction is nothing except choosing
objects and checking whether they are extensions of the word.

Wittgenstein says the meaning of a word is identified by   how it
is used.
\begin{quote}
``For a large class of cases-though not for all- in which we
employ the word `meaning' it can be defined thus: the meaning of a
word is its use in the language"~\footnote{The quote is written
from Stanford Encyclopeida of Philosophy~\cite{kn:plato}.}
\end{quote}
The use of a word happens in time, therefore we may say that the
meaning of a word exists in time, and it is an \emph{unfinished
entity} similar to a choice sequence~\cite{kn:ob}. The meaning of
a word is not predetermined, and as time passes the word finds its
meaning through the interaction of the subject with the language.
  The meaning  is a dynamic temporal (mental) construction~\footnote{An object is
temporal exactly if it exists in time, and it is dynamic if at
some moments are part added to it or removed from it~\cite{kn:BH},
page~16. Another philosophical framework that knows possible the
evolution of a meaning is ``dynamic Semantics" (see
http://plato.stanford.edu/entries/dynamic-semantics/). In this
framework, meaning is context change potential.}. In Brouwer's
intuitionism, mathematical objects are mental construction.
Brouwer's choice sequences, as a kind of mathematical objects, are
dynamic temporal objects~(see~page~16,~\cite{kn:BH}). The meaning
of a word  (which is identified via how it is used by the human
being) has lots of common with   a choice sequence. At each stage
of time, the human being only experienced a finite set of things
that whether they are extensions of the word or not. Also, in a
choice sequence, at each stage of time only a finite segment of
the sequence is determined. As the human being freely chooses
another thing  to check its extension status for the word, the
meaning of the word may persistently change (the construction of
the human's brain (or mind) may persistently change). It is
similar to the act of the human being in developing a choice
sequence.

 The use of
a word is dependent to the human being and the way that he uses
(interacts  via) the word in (with) the environment.  We may
assume that the human being has freedom to choose things in any
order that he wants to check their status of being extensions of a
word. The different order of choosing objects may cause that the
meaning of the word evolves in different ways. But since the human
being cannot go back to the past, he   just lives in one way of
evolution. As soon as the human being chooses an object and checks
whether it is an extension of a word $w$, (the biological
construction of) his mind may persistently evolve, and the meaning
of the word $w$ persistently changes. Therefore, it seems possible
that the intension of  a word would be persistently evolutionary,
and

\begin{quote}
the interaction of the human being with the language may make the
meaning of a word persistently evolve. Assuming  the free will for
the human being, the behavior of the human being is not
predetermined, and as a consequence  the meaning of a word needs
not to be predetermined.
\end{quote}

\begin{itemize}
\item[\textbf{Q2)}] Is it possible for a subject to distinguish
between persistently evolutionary intensions and static ones? In
other words, is it possible for a subject to determine that
whether the intension of a word is static or persistently
evolutionary?
\end{itemize}
The answer is No. It is not possible for the human being to
recognize whether the meaning of a word, in the course of his
thinking  activities, persistently evolves or remains constant.
Suppose $w$ be a word. Two cases are possible
\begin{itemize}\item[1)] the intension of the word $w$ is static  and
\item[2)] the intension of the word $w$ is persistently
evolutionary.
\end{itemize} In both cases, at each stage of time, the human being
just has experienced the status of extension of a finite set of
objects. The human being just has access to his pervious
experiences and does not have access to the future. So at each
stage of time,  all information that the human being has
 about the word $w$ is a finite set of objects
 $\{d_1,d_2,...,d_n\}$ that their extension status are determined. This information
 of the word $w$ are the same in both cases.

 The human being cannot differ between these two cases based on his
 obtained information. The persistent evolution is similar to
 being static in view of past experiences. The difference of the
 persistent
 evolution and being static is in future. But the human being does
 not have access to the future.  As soon as, the meaning of a word
 evolves then it has been evolved and the subject cannot
 go back to the past and experience  another way of evolution.


\begin{example} Suppose that I am in a black box with two windows:
an input window and an output one. You give natural numbers as
input to the black box and receive a natural numbers as output of
the black box. I do the following strategy in the black box. I
plan to output $1$ for each input before You give the black box
$5$ or $13$ as inputs. If you give $5$ as input (and you have not
given $13$ already) then   after that time, I output $2$ for all
future inputs that have not been already given to the black box.
For those natural numbers that you have already given them as
input, I still output the same $1$.
 If you give $13$ as input (and you have not given
$5$ already) then   after that time, I output $3$ for all future
inputs that have not been already given to the black box. For
those natural numbers that you have already given them as input, I
still output the same $1$.

\end{example}
The black box of the above example behaves well defined. But it
persistently evolve through interactions with the environment, and
the function that the black box provides is not a predetermined
function. If one does not have access to the inner structure of
the black box, he could always assume that  there exists a
 a static machine in the black box.

\begin{remark} \textsc{Deterministic vs. Predetermination}.
Being deterministic does not force to be predetermined. The human
may  computationally intend  a function deterministically, but it
is not needed that the language to be predetermined. It may be
determined as time passes by the free will of the human being.
\end{remark}
We propose a postulate about the extension of a word as follows:
Suppose $w$ is a word and $\mathcal{O}$ is a class of objects. We
refer to the set of extension of $w$  by $\mathrm{E}(w)\subseteq
\mathcal{O}$. Our proposed postulate which we call it "the
Postulate of Persistent Evolution", $\mathbf{PPE}$,  says:

\begin{itemize}
\item[$\mathbf{PPE}$:] if a subject  has not yet   proved that
$\mathrm{E}(w)$ is finite (or in other words, if a subject has not
yet listed all the element of $\mathrm{E}(w)$ on paper) then he
could not yet disprove that the meaning (intension) of $w$ does
not persistently evolve.
\end{itemize}

Suppose that a subject wants to prove that \begin{itemize}
\item[i.] the meaning of a word $w$ is static and does not
persistently evolve, and  \item[ii.] the set $\mathrm{E}(w)$ is
predetermined and does not depend to the order that he chooses
objects from the domain $\mathcal{O}$ to check whether they are
extensions of $w$ or not.\end{itemize}

The subject at each stage of time, only knows the status of a
finite number of objects in $\mathcal{O}$ that whether they are
extensions of $w$ or not. Suppose $E(w)$ is infinite. Then the
subject has never written all extensions of $w$ at any stage of
time. He always could know it possible that the the meaning of the
word $w$ may persistently change. But since this change happens
persistently,   he cannot recognize whether the meaning is static
or not, based on the finite history that he has access to it.

\begin{quote}If a subject does not sense a change about a process, then
he may  (wrongly) presuppose that the process is static and
independent of his interaction with the process. In spite of this,
in the case that a process persistently evolves, the subject does
not sense any change as well!  We only sense a change whenever we
discover that an event which has been sensed before is not going
to be sensed similar to past. Persistent evolution always respects
the past. As soon as a subject experiences an event, then whenever
in future he examines the same event, he will experience it
similar to past. Persistent evolution effects the future which has
not been determined yet. \end{quote}

In other words,  the postulate $\mathbf{PPE}$ says that

\begin{center} it is not possible for a subject to differ between static
intensions and persistently evolutionary one. \end{center}
\section{Persistently Evolutionary Semantics}\label{SEM}

 In this
section, to clarify the notion of persistently evolutionary
intensions, we introduce a kind of Kripke structures  that we name
Persistently Evolutionary Kripke structures.

 Let $P=\{p_i\mid i\in I\}$ be
a   set of atomic propositional formulas for an index set
$I\subseteq \mathbb{N}$, and $A=\{a_i\mid i\in I'\}$ ($I'\subseteq
\mathbb{N}$ is an index set) be a set that is assumed as the set
of actions of an agent $ag$.

\begin{definition}
A Persistently Evolutionary Kripke structure  over a set of
actions $A$ and a set of atomic formulas $P$ is a tuple $K=\langle
S,\Pi=\{\pi_j\mid j\in J\}, \sim_{ag} , V\rangle$ where
$J\subseteq \mathbb{N}$ is an index set, and

\begin{itemize}

\item $S=A^*\times J$  is the set of all possible worlds ($A^*$ is
the set of all finite sequences of actions in $A$). For each
$s=(\vec{x},i)\in S$, we call $i$ the \emph{meaning index} of the
state $s$.

\item $\Pi$ is the set of \emph{meaning functions}. Each $\pi_i\in
\Pi$ ($i\in J$), is a partial function from $S\times A$ to
$\mathbb{F}P$ (the set of finite subsets of $P$) that its domain
is $(A^*\times\{i\})\times A$.  To each state $s=(\langle
b_1,b_2,...,b_n\rangle,i)$, and each action $a\in A$ the function
$\pi_i$ corresponds  a finite subset of $P$ as the meaning of the
action $a$, satisfying the following condition (\emph{persistently
evolutionary condition}):

for each state $s=(\langle b_1,b_2,...,b_n\rangle,i)\in S$, if an
action $a$ is appeared  in the finite sequence $\langle
b_1,b_2,...,b_n\rangle$, and $\langle
b_1,b_2,...,b_n\rangle=\langle b_1,b_2,...,b_i\rangle\langle
a\rangle\langle b_{i+2},...,b_n\rangle$, then
$\pi_i(s,a)=\pi_i((\langle b_1,b_2,...,b_i\rangle,i),a)$.

\item The agent $ag$ is an operator  which chooses actions from
$A$ and performs them. By his operation, it makes the universe
evolve. If $s=(\vec{x},i)\in S$ is the current state of the model
$K$ and $ag$ performs $a\in A$, then the current world evolves to
$ s'=(\vec{x}.\langle a\rangle,i)$. Note that via evolution, the
meaning index of the states does not change. We say the agent $ag$
lives in the meaning function $\pi_i$, or  in other words, the
actual meaning function for the agent $ag$ is $\pi_i$.

\item $V$  is a function form $S$ to $2^P$ defined as follows: for
each $s\in S$, $s=(\langle b_1,b_2,...,b_n\rangle,j)$,
\begin{center}$V(s)=\pi_j((\langle\rangle,j), b_1)\cup(\bigcup_{1\leq
i\leq n} \pi_j((\langle b_1,b_2,...,b_i \rangle,j), b_{i+1}))$.
\end{center}

\item $\sim_{ag}\subseteq S\times S$ is a binary relation which
satisfies the following condition: for all two states $s_1,s_2\in
S$, we have $s_1\sim_{ag} s_2$ whenever $s_1=(\vec{x},i)$ and
$s_2=(\vec{x},j)$ for some $\vec{x}=\langle
x_1,x_2,...,x_k\rangle$ and $i,j\in J$ such that for all $1\leq
t\leq k$, $\pi_i((\langle x_1,...,x_{t-1}\rangle ,i),
x_t)=\pi_j((\langle x_1,...,x_{t-1}\rangle ,j), x_t)$.

\end{itemize}

\end{definition}

The relation $\sim_{ag}$ is an indistinguishability  relation for
the agent $ag$. If $s_1\sim_{ag} s_2$ then it means that the agent
$ag$ cannot distinguish between these two states, since  all
experiences that he has observed in both states are the same.

\begin{definition}\label{osens} Let $K=\langle S,\Pi=\{\pi_j\mid j\in J\}, \sim_{ag} , V\rangle$ be
a persistently evolutionary Kripke Structure. We say a meaning
function $\pi_i\in \Pi$ is static if it is not order-sensitive.
That is, for every $n\in \mathbb{N}$, for every $a_1,a_2,...,a_n,
a\in A$, for every permutation $\delta:
\{1...n\}\rightarrow\{1...n\}$,

\begin{center} $\pi_i((\langle
a_{\delta(1)},a_{\delta(2)},...,a_{\delta(n)}\rangle,
i),a)=\pi_i((\langle a_{1},a_{2},...,a_{n}\rangle,
i),a)$.\end{center}

\end{definition}

\begin{notation}
For each state $s=(\vec{x},i)$, we let  $D(s)=\{s'\mid
s'=(\vec{y},i), \vec{x}~is~a~prefix~of~\vec{y}\}$.
\end{notation}
\begin{definition}\label{indist}  Let $K=\langle S,\Pi=\{\pi_j\mid j\in J\}, \sim_{ag} , V\rangle$ be
a persistently evolutionary Kripke Structure. Suppose
$s=(\vec{a},i)$, $\vec{a}\in A^*$, and $i\in J$ be a current state
that the agent $ag$ lives in. We may say that the agent $ag$ can
never become conscious that whether his world is static or
persistently evolutionary whenever   for every $s'=(\vec{a},i)\in
D(s)$ there exists a meaning function $\pi_j\in \Pi$ which is not
static, and for $s''=(\vec{a},j)$, we have $s'\sim_{ag}s''$.
\end{definition}

\begin{definition} Let $P$ be a non-empty set of propositional
variables. The language $L(P)$ is the smallest superset of $P$
such that
\begin{center}
if $\varphi,\psi\in L(P)$ then $\neg \varphi,\
(\varphi\wedge\psi), (\varphi\vee\psi), (\varphi\rightarrow\psi),
K_{ag}\varphi,\Box \varphi, C_f\varphi\in L(P), $,
\end{center}
  $C_f\varphi$ has to be read as ``the formula $\varphi$ conflicts with the free will of the agent
  $ag$",
  $K_{ag}\varphi$ has to be read as ``the agent $ag$ knows $\varphi$", and
  $\Box\varphi$ has to be read as ``$\varphi$ is necessary true".
\end{definition}

\begin{notation}Let $K$ be a Kripke model with the set of state $S$. For each
subset $A\subseteq S$, $K_A$ is defined to be the same Kripke
model $K$ which its set of states is restricted to the set $A$.
\end{notation}

\begin{definition}
In order to determine whether a formula $\varphi\in L(P)$ is true
in a current world $(K,s)$, denoted by $(K,s)\models \varphi$, we
look at the structure of $\varphi$:
\[
\begin{array}{l}\begin{array}{ccccc}
  (K,s)\models p & \emph{iff} &p\in V(s) \\
(K,s)\models
(\varphi\vee\psi) & \emph{iff} &  (K,s)\models\varphi~or~(K,s)\models\psi \\
(K,s)\models (\varphi\rightarrow \psi) & \emph{iff} & for~all~
t\in D(s),~ if
(K,t)\models\varphi~ then~(K,t)\models\psi \\
  (K,s)\models
(\varphi\wedge\psi) & \emph{iff} & (K,s)\models\varphi~and~(K,s)\models\psi \\
  (K,s)\models\neg\varphi & \emph{iff} & for~all~ t\in D(s),~(K,t)\not\models\varphi \\
 (K,s)\models
\Box\varphi & \emph{iff} &
for~all~ t\in D(s),~(K,t)\models\varphi \\

(K,s)\models K_{ag}\varphi & \emph{iff} & for~all~ t\in S,
~if~t\sim_{ag} s~then~(K,t)\models\varphi\\

(K,s)\models C_f\varphi & \emph{iff} &
there~exsits~an~infinite~set~ Path=\{
s_0,s_1,...\},\\~&~&~where~s_0=s,~and ~for ~each~i,~ s_{i+1}\in
D(s_i)~and~(K_{Path},s_i)\not\models \varphi

\end{array}\quad\quad\quad\quad\quad\quad\quad\quad\quad\quad\quad\quad\quad
\end{array} \]
 The current state
of the Kripke model $K$ is not a fixed state. The current state
evolves due to  $agent$'s operation, and
 it is not possible for the $agent$ to travel back in time from a
 state $s$ to one of its prefixes. Note that during the
evolution, the meaning function  does not change. That is, if the
current state is $s=(\vec{x},i)$ and due to executing an action
$a$, the current state changes to be $s'$, then
$s'=(\vec{x}.\langle a\rangle,i)$ for the same $i$. In this case,
we   call $\pi_i$ the actual meaning function of the universe.

 The
semantics of $C_f\varphi$ says that the agent $ag$ can interact
with the universe  and evolve it in a way that never $\varphi$
holds true. Therefore, the assumption of truth of $\varphi$
conflicts with the free will of the agent.
\end{definition}

\begin{definition} Let $K=\langle S,\Pi,\sim_{ag}, V \rangle$ be a persistently evolutionary Kripke structure.
We say
an action $a\in A$, at the state $s=(\vec{x},i)\in S$, is a static
action whenever for all $s_1,s_2\in D(s)$ , we have
$\pi_i(s_1,a)=\pi_i(s_2,a)$. That is, if the $agent$ starts   from
the state $s$ to perform actions, then the different orders that
he may perform the actions does not make the meaning of the action
`$a$' change.
\end{definition}

\begin{remark}
At each state, the $agent$ cannot go back to past to experience
his universe in different ways, thus he cannot distinguish between
static actions and persistently evolutionary ones.
\end{remark}
\begin{example}
Suppose $A=\mathbb{N}$ as a set of actions, and $P=\{p_{i,j}\mid
i,j\in \mathbb{N}\}$ as a set of atomic propositions.  For each
finite sequence of numbers $\vec{x}=\langle
x_1,x_2,...,x_n\rangle$, and $y\in A$,
\begin{itemize}
\item[] if for all $1\leq i\leq n$, $x_i\neq 5$,$x_i\neq 13$,
define $\pi(\vec{x}, y)=\{p_{y,1}\}$

\item[]  if for some $1\leq i\leq n$, $x_i= 5$ and for all $j<i$
$x_j\neq 13$ then
\begin{itemize}
\item if for some $t\leq\min(\{i|x_i=5\})$, $y=x_t$ then define
$\pi(\vec{x}, y)=\{p_{y,1}\}$ else define $\pi(\vec{x},
y)=\{p_{y,2}\}$.
\end{itemize}
\item[]  if for some $1\leq i\leq n$, $x_i= 13$ and for all $j<i$
$x_j\neq 5$ then
\begin{itemize}
\item if for some $t\leq\min(\{i|x_i=13\})$, $y=x_t$ then define
$\pi(\vec{x}, y)=\{p_{y,1}\}$ else define $\pi(\vec{x},
y)=\{p_{y,3}\}$.
\end{itemize}
\end{itemize}
The function $\pi$ satisfies the persistently evolutionary
condition. We have $(K,\langle 1,3\rangle)\models  p_{1,1}$. Also
$(K,\langle 1,3\rangle)\models C_f  p_{6,1}$, and $(K,\langle
1,3\rangle) \models  C_f  \neg p_{6,1}$.  The value of $p_{6,1}$
is not predetermined yet and depends on the free will of the
agent.
\end{example}

One may check that the indistinguishability relation $\sim_{ag}$
is
\begin{itemize}

\item[$1)$] \emph{reflexive} (for all $s\in S$, $s\sim_{ag}s$);

\item[$2)$] \emph{transitive} (for all $s,t,u\in S$, if
$s\sim_{ag}t$ and $t\sim_{ag}u$ then $s\sim_{ag}u$);

\item[$3)$] \emph{Euclidean} (for all three
 states $s,t,u\in S$
 if $s\sim_{ag}t$ and $s\sim_{ag}u$ then
$t\sim_{ag}u$).
\end{itemize} Therefore, persistently evolutionary Kripke structures are models
for the standard epistemic logic $S5$~\cite{kn:dit3} which
consists of axioms $A1-A5$ and the derivation rules $R1$ and $R2$
given below

\[
\begin{array}{l}\emph{A1:~ Axioms~of~propositional~logic}\quad\quad\quad\quad\quad
\quad\quad\quad\quad\quad\quad\quad\quad\quad\quad\quad\quad\\
\emph{A2:~} (K\varphi\wedge
K(\varphi\rightarrow\psi))\rightarrow K\psi\\
\emph{A3:~} K\varphi\rightarrow\varphi\\
\emph{A4:~} K\varphi\rightarrow KK\varphi\\
\emph{A5:~} \neg K\varphi\rightarrow K\neg K\varphi\\
\end{array} \]

\[
\begin{array}{l}\emph{R1:~}
\vdash\varphi,\ \vdash\varphi\rightarrow\psi\Rightarrow\ \vdash\psi\\
\emph{R2:~}\vdash\varphi\Rightarrow K\varphi,
\quad\quad\quad\quad\quad\quad\quad\quad\quad\quad\quad\quad\quad\quad\quad\quad\quad\quad\quad
\end{array} \]

\subsection{A Kripke Model for Persistently Evolutionary
Intensions}

 Now we describe the notion of persistently evolutionary
intensions using persistently evolutionary Kripke models.

Let $\textsc{Language}=\{w_1,w_2,...\}$ be a set of words for a
subject $IA$, and $X=\{x_1,x_2,...\}$ be an infinite set of
objects that could be assumed as possible extensions of words in
$\textsc{Language}$.

The subject chooses a word $w\in \textsc{Language}$ and an object
$x\in X$ to check whether $x$ is an extension of the word $w$ or
not. Therefore, the set of actions of the Kripke model is defined
to be $A_e=\{(w_i,x_j)\mid i,j\in \mathbb{N}\}$. The set of atomic
propositions is defined to be $P_e=\{p_{(w_i,x_j,0)}\mid i,j\in
\mathbb{N}\}\cup\{p_{(w_i,x_j,1)}\mid i,j\in \mathbb{N}\}$.

The agent $IA$ chooses a word $w_i$ and an object $x_j$ to check
whether $x_j$ is an extension of the word $w_i$. If at state $s$,
he chooses $(w_i,x_j)$ then the current state evolves to $s'=
s.\langle(w_i,x_j)\rangle$.

 We let the set of  meaning functions
$\Pi_e$ to be the set of all functions   $\pi_i$s which satisfy
the following conditions:

\begin{itemize}

\item[1-] For each state $s=(\vec{x},i)$, and action $(w_i,x_j)$
 either $\pi_i(s,(w_i,x_j))=p_{(w_i,x_j,1)}$
(we read it as ``at the current state $s$, the agent $IA$ checked
that whether $x_j$ is an extension of the word $w_i$ and found out
the answer `yes') or $\pi_i(s,(w_i,x_j))=p_{(w_i,x_j,0)}$ (we read
it as ``at the current state $s$, the agent $IA$ checked that
whether $x_j$ is an extension of the word $w_i$ and found out the
answer `no').

\item[2-] Each $\pi_i\in \Pi$ satisfies  the persistently
evolutionary condition.

\end{itemize}

We call the Kripke model $K_e=(S_e,\Pi_e,\sim_{ag},V_e)$
(introduced above) the model of persistently evolutionary
intensions.

\begin{definition}\label{peint} We say the intension of a word $w\in \textsc{Language}$ is
static (or its extension is not order-sensitive)
 at a state $s=(\vec{x},j)$ whenever for all
$a\in\{(w,x_i)\mid i\in \mathbb{N}\}$ and for all  $s_1$ and $s_2$
in $D(s)$, we have $\pi_j(s_1,a)=\pi_j(s_2,a)$.
\end{definition}

\begin{theorem} For each state $s=(\vec{x},j)\in S_e$ and each
word $w\in \textsc{Language}$ there exist two states
$s_1=(\vec{x},t)\in S_e$ and $s_2=(\vec{x},t)\in S_e$ such that
$s\sim_{ag} s_1$ and $s\sim_{ag} s_2$, and the intension of $w$ is
static at $s_1$, and persistently evolutionary at $s_2$.
\end{theorem}\begin{proof} The proof is straightforward.
\end{proof}

The above theorem says that it is not possible for the agent who
lives in the persistently evolutionary Kripke model $K_e$ to gets
aware that whether the intension of a word $w$ is static or
persistently evolutionary. It is because, at each stage of time
 (at each state of the Kripke model $K_e$) the agent only
observed a finite set of experiences, and as he cannot travel to
the past (go  back to a prefix of the current state), he cannot
experience  different orders of his behavior to be assure that if
the actual meaning function which the universe evolves in, is
order-sensitive or not.

\begin{quote}If the   agent wants to be aware of a change, then
he must experience an event different from the way that he has
experienced the same event already. But as the evolution happens
persistently, it is impossible.\end{quote} In persistently
evolution, the behavior of the agent changes the future which has
not yet occurred.

\section{Persistently Evolutionary Semantics for Propositional
Logic}\label{Logic}

Persistently evolutionary Kripke structures can be considered as
models for propositional logic.  Let $P_l=\{p_i\mid i\in I\}$ be a
set of atomic formulas. The language $L_l$ of propositional logic
is the smallest set containing $P_l$ satisfying the  following
condition:
\begin{center}
$\phi,\psi\in L\Rightarrow \varphi\wedge\psi, \varphi\vee\psi,
\neg\varphi, \varphi\rightarrow\psi, K\varphi, C_f\varphi,
\Box\varphi \in L_l$.
\end{center}
We say $K=\langle S,\Pi,\sim_{ag}, V \rangle$ is a  Kripke
structure for propositional logic whenever the set of actions is
$A=P_l$, the set of atomic formulas of the structure $K$ is
$\mathrm{P}=\{p=b\mid p\in P_l, b\in\{0,1\}\}$, and for each
$\pi_i\in\Pi$, $\pi_i((\langle
p_1,p_2,...,p_n\rangle,i),p)=\{p=b\}$ for some $b\in\{0,1\}$.

\begin{definition} For every formula $\varphi\in L_l$, we define
\[
\begin{array}{l}\begin{array}{ccccc}
  (K,s)\models p & \emph{iff} &p=1\in V(s) \\
(K,s)\models
(\varphi\vee\psi) & \emph{iff} &  (K,s)\models\varphi~or~(K,s)\models\psi \\
(K,s)\models (\varphi\rightarrow \psi) & \emph{iff} & for~all~
t\in D(s),~ if
(K,t)\models\varphi~ then~(K,t)\models\psi \\
  (K,s)\models
(\varphi\wedge\psi) & \emph{iff} & (K,s)\models\varphi~and~(K,s)\models\psi \\
(K,s)\models\neg\varphi & \emph{iff} & for~all~ t\in D(s),~(K,t)\not\models\varphi \\
(K,s)\models K_{ag}\varphi & \emph{iff} & for~all~ t\in S,
~if~t\sim_{ag} s~then~(K,t)\models\varphi\\

K,s)\models \Box\varphi & \emph{iff} &
for~all~ t\in D(s),~(K,t)\models\varphi \\

(K,s)\models C_f\varphi & \emph{iff} &
there~exsits~an~infinite~set~ Path=\{
s_0,s_1,...\},\\~&~&~where~s_0=s,~and ~for ~each~i,~ s_{i+1}\in
D(s_i)~and~(K_{Path},s_i)\not\models \varphi

\end{array}\quad\quad\quad\quad\quad\quad\quad\quad\quad\quad\quad\quad\quad
\end{array} \]
\end{definition}

One may check that if we omit the operator $K$, $\Box$, and $C_f$
from the language $L_l$ then the persistently evolutionary
semantics is sound and complete for intuitionistic propositional
logic (see Chapter 2, \cite{kn:TD}).

\section{Persistently Evolutionary Semantics for Predicate
Logic}\label{Logic2} In this part, we propose a persistently
evolutionary semantics for predicate logic.

A predicate language  $L_o$  contains
\begin{itemize}
\item  a set of predicate symbols $\mathcal{R}$, and a natural
number $n_R$ for each $R\in \mathcal{R}$ as its ary,

\item  a set of function symbols $\mathcal{F}$, and a natural
number $n_f$ for each $f\in \mathcal{F}$,

\item a set of constant symbols $\mathcal{C}$.

\end{itemize}

\begin{definition}
A \emph{partial} $L_o$-structure $\mathcal{N}$  is given by the
following data
\begin{itemize}
\item[1)] a nonempty set $N$ called the domain,

\item[2)] a partial function $f^\mathcal{N}:N^{n_f}\rightarrow N$,
for each $f\in \mathcal{F}$,

\item[3)] a set $R^\mathcal{N}\subseteq N^{n_R}$ for each $R\in
\mathcal{R}$,

\item[4)] a \emph{partial} zero-ary function $c^\mathcal{N}\in N$
for each $c\in C$. (In this way, there could be some constant
symbols $c\in C$, which are not interpreted in the structure.)
\end{itemize}
\end{definition}
We refer to $R^\mathcal{N},f^\mathcal{N},c^\mathcal{N}$ as
interpretations of   symbols $R,f,c$.

\begin{definition}
 TERM is the smallest set containing

\begin{itemize}
\item[] variable symbols,

\item[] constants  symbols in $C$,

\item[]  for each function symbol $f\in \mathcal{F}$, if
$t_1,t_2,...,t_{n_f}\in TERM$ then $f(t_1,t_2,...,t_{n_f})$ is a
term.
\end{itemize}

\end{definition}
 The interpretation of a term $t$, denoted by $t^\mathcal{N}$ is
 defined to be  a \emph{partial }function from $N^k$ to $N$ for some
 $k$, similar to the interpretation of terms in model theory (see
 definition~1.1.4~of~\cite{kn:model}). The only difference is that
 the interpretations are partial functions.

\begin{definition}
FORMULA is the smallest set satisfying the following conditions:
\begin{itemize}
\item[] $\perp\in Formula$,

\item[] $t_1,t_2\in TERM$ then $t_1=t_2\in FORMULA$,

\item[] for each predicate  symbol $R\in \mathcal{R}$, if
$t_1,t_2,...,t_{n_R}\in TERM$ then $R(t_1,t_2,...,t_{n_R})\in
FORMULA$

\item[] $\varphi,\psi\in FORMULA$ then $\neg\varphi,
\varphi\wedge\psi, \varphi\vee\psi, \varphi\rightarrow\psi,
\forall y\varphi, \exists y\varphi \in FORMULA$.
\end{itemize}

\end{definition}
A persistently evolutionary Kripke structure $K_{L_o}$ for the
language $L_o$ is defined as follows:
\begin{itemize}
\item[-] a nonempty set $\mathcal{O}$ called the domain,

\item[-] The set of actions of the Kripke structure is

$A_O=\{R(\vec{o})\mid \vec{o}\in \mathcal{O}^{n_R},~R\in
\mathcal{R}\}\cup$

$ \{f(\vec{o})\mid \vec{o}\in \mathcal{O}^{n_f},~f\in \mathcal{F}
\}\cup$

 $\{c\mid  c\in \mathcal{C}\}$.

\item[] The set of atomic propositions of the Kripke structure is

$P_O=\{(R(\vec{o})=b)\mid b\in\{0,1\},\vec{o}\in \mathcal{O}^{n_R}, R\in \mathcal{R} \}\cup$\\
 $\{(f(\vec{o})=o')\mid \vec{o}\in \mathcal{O}^{n_f},
o'\in \mathcal{O}, f\in \mathcal{F}\}\cup$\\
$\{(c_j=o)\mid c\in \mathcal{C}, o\in \mathcal{O}\}$.
\end{itemize}

 The set of meaning functions $\Pi$ of the Kripke structure is
the set of all functions $\pi_i$, which satisfy persistently
evolutionary condition and
\begin{itemize}
\item[] $\pi(s,R(\vec{o}))=\{(R(\vec{o})=b)\}$ for some $
b\in\{0,1\}$, and

\item[] $\pi(s,f(\vec{o}))=\{(f(\vec{o})=o')\}$ for some $o'\in
\mathcal{O}$.

\item[]$\pi(s,c)=\{(c=o)\}$, for some $o\in \mathcal{O}$.
\end{itemize}

The meaning of predicates and the value of functions  are not
predetermined in the Kripke structure. As soon as the agent $ag$
chooses a predicate symbol $R$ and a tuple $(o_1,o_2,...,o_{n_R})$
to find the value of $R(o_1,o_2,...,o_{n_R})$, the meaning
function gives out an atomic proposition
$(R(o_1,o_2,...,o_{n_R})=b)$, $b\in\{0,1\}$, and the current state
evolves to a new state.

\begin{definition}
Let $s$ be a state of the persistently evolutionary Kripke model
$K_{L_o}$. The partial $L_o$-structure of  the state $s$, denoted
by $\mathcal{N}_s$, is defined as follows:
\begin{itemize}
\item[1-] the domain of the structure $N_s$ is the same domain of
the  Kripke model $\mathcal{O}$.

\item[2-] For each symbolic predicate  $R\in \mathcal{R}$, the
relation $R^{\mathcal{N}_s}$ is defined to be $\{\vec{o}\mid
(R(\vec{o})=1)\in V(s)\}$.

\item[3-] For each symbolic function $f\in \mathcal{F}$, the
partial  function $f^{\mathcal{N}_s}$ is defined to be
$\{(o,o')\mid (f(o)=o')\in V(s)\}$.

\item[4-] For each symbolic constant $c\in \mathcal{C}$, we define
$c^{\mathcal{N}_s}=o$ if $(c=o)\in V(s)$.
\end{itemize}
\end{definition}
 We say a constant $c$ is
predetermined at a state $s$
 whenever for some $o\in \mathcal{O}$, $c^{\mathcal{N}_s}=o$.
 We say a predicate symbol $R$ is predetermined for
$\vec{o}$ at a state $s$, whenever for some $b\in\{0,1\}$,
$(R(\vec{o})=b)\in V(s)$. We say a function symbol $f$ is
predetermined for
 $\vec{o}$ at a state $s$, whenever for some $o'\in \mathcal{O}$,
 $f^{\mathcal{N}_s}(\vec{o})=o'$. We simply can inductively define being
 predetermined for terms and formulas.

\begin{definition} Let $K_{L_o}=\langle S,\Pi=\{\pi_i\mid i\in I\},
V\rangle$ be a persistently evolutionary Kripke structure for the
language $L_o$. Let $s$ be a state of this model. Also let $\phi$
be a formula with free variables $\vec{y}=(y_1,y_2,...,y_n)$, and
let $\vec{o}=(o_1,o_2,...,o_n)\in \mathcal{O}^n$. We inductively
define $(K,s)\models \varphi(\vec{o})$ as follows.

\begin{itemize}
\item[-] if $\phi$ is $t_1=t_2$, then $(K,s)\models \phi(\vec{o})$
iff $t_1^{\mathcal{N}_s}(\vec{o})=t_2^{\mathcal{N}_s}(\vec{o})$,

\item[] if $\phi$ is $R(t_1,t_2,...,t_{n_R})$, then $(K,s)\models
\phi(\vec{o})$ iff $(t_1^{\mathcal{N}_s}(\vec{o}),
t_2^{\mathcal{N}_s}(\vec{o}),...,t_{n_R}^{\mathcal{N}_s}(\vec{o}))\in
R^{\mathcal{N}_s}$,

\item[] if  $\phi$ is $\neg\psi$, then $(K,s)\models
\phi(\vec{o})$ iff for all $w\in D(s)$, $w\not\models
\psi(\vec{o})$,

\item[]  if  $\phi$ is $\varphi\rightarrow \psi$, then
$(K,s)\models \phi(\vec{o})$ iff for all $w\in D(s)$, $w\models
\varphi(\vec{o})$, then $w\models \psi(\vec{o})$,

\item[] if  $\phi$ is $\varphi\wedge \psi$, then $(K,s)\models
\phi(\vec{o})$ iff   $s\models \varphi(\vec{o})$, and $s\models
\psi(\vec{o})$,

\item[] if  $\phi$ is $\varphi\vee \psi$, then $(K,s)\models
\phi(\vec{o})$ iff   $s\models \varphi(\vec{o})$, or $s\models
\psi(\vec{o})$,

\item[]   if  $\phi$ is $\forall x \psi(\vec{y},x)$, then
$(K,s)\models \phi(\vec{o})$ iff for all $w\in D(s)$, for all
$o'\in \mathcal{O}$, if $\psi(\vec{o},o')$ is defined
 in the partial structure
$\mathcal{N}_w$ (\underline{predetermined} at the state $w$) then
$w\models \psi(\vec{o},o')$,

 if  $\phi$ is $\exists x \psi(\vec{y},x)$, then
$(K,s)\models \phi(\vec{o})$ iff there exists $o'\in \mathcal{O}$,
$s\models \psi(\vec{o},o')$.
\end{itemize}
\end{definition}

\begin{proposition}
For every formula $\varphi\in L_o$, and every state $(K,s)$, if
$(K,s)\models \varphi$ then for all $s'\in D(s)$, $(K,s')\models
\varphi$.
\end{proposition}\begin{proof} It is straightforward. \end{proof}

\subsection{Free Will}
One of our purpose of proposing persistently evolutionary
semantics is to provide a framework to formalize the notion of
free will. We discussed the notion of free will in section~4.3
of~\cite{kn:comp1}. In this part, we repeat the same discussion
using persistently evolutionary Kripke structures.

Let $R$ be a a one-ary predicate symbol. Let
$\mathcal{O}=\{0,1\}^*$ be the set of all finite strings over $0$
and $1$. Consider the meaning function $\pi_j$ as follows: for
$s=(\langle R(x_1), R(x_2),...,R(x_k)\rangle,j)$, and $x\in
\{0,1\}^*$,
\begin{itemize}
\item if for some $1\leq i\leq k$, $x=x_i$  then  $\pi_j(s,R(x))$
is defined to be  $\pi_j(s',R(x_i))$ for $s'=(\langle R(x_1),
R(x_2),...,R(x_{i-1})\rangle,j)$,


\item if for all $1\leq i\leq k$, $x\neq x_i$, and  there exists
$1\leq i\leq k$, such that $x_i=x0$ or $x_i=x1$ and for
$s'=(\langle R(x_1), R(x_2),...,R(x_{i-1})\rangle,j)$,
$\pi_j(s',R(x_i))=\{R(x_i)=1\}$ then $\pi_j(s,R(x))$ is defined to
be $\{ (R(x)=0)\}$,

\item otherwise, $\pi_j(s,R(x))$ is defined to be $\{ (R(x)=1)\}$,
\end{itemize}
It is easy to check that the meaning function $\pi_j$ behaves
similar to the persistently evolutionary Turing machine $PT_1$
introduced in example~4.6 in~\cite{kn:comp1}. The next theorem is
a formal version of the theorem~4.9 in~\cite{kn:comp1}. Let $K'$
be the persistently evolutionary Kripke model which the set of its
meaning function $\Pi$ is $\{\pi_j\}$.

\begin{theorem}\label{fr} Let
\begin{center}
$\varphi:= (\exists k\in \mathbb{N})(\forall n>k)(\exists
x\in\{0,1\}^*)(|x|=n\wedge R(x))$.
\end{center} We have for the initial state $s=(\langle\rangle,j)$, \begin{center} $(K',s)\models \Box C_f \varphi
\wedge \Box C_f \neg \varphi$.\end{center}
\end{theorem}
\begin{proof} The agent can develop the future in two ways such
if  the first way happens $\varphi$ is true in the universe,  but
if the second happens $\neg\varphi$ is true.

We define two ordering $\preceq_1,\preceq_2$ on the elements of
$\{0,1\}^*$ as follows. Let $x_1,x_2\in\{0,1\}^*$.
\begin{itemize}
\item[1-] if $|x_1|<|x_2|$ then $x_1\preceq_1 x_2$,

\item[2-] if  $|x_1|=|x_2|$ then \begin{itemize} \item[]
$0\preceq_1 1$,

\item[] if $x_1\preceq_1 x_2$ then $x_1a\preceq_1 x_2a$, for $a\in
\{0,1\}$,

\item[] $x_10\preceq_1 x_11$.
\end{itemize}
\end{itemize}, and

\begin{itemize}
\item[1-] if $|x_1|+1<|x_2|$ then $x_1\preceq_2 x_2$,

\item[2-] if  $|x_1|=|x_2|$ then \begin{itemize} \item[]
$x_1\preceq_2 x_2$ iff $x_1\preceq_1 x_2$,
\end{itemize}
\item[3-] if $|x_1|+1=|x_2|$ and $|x_1|$ is even then
$x_1\preceq_2 x_2$,

\item[4-] if if $|x_1|+1=|x_2|$ and $|x_1|$ is odd then
$x_2\preceq_2 x_1$,
\end{itemize}

Now let $y_1,y_2,...$ be an enumeration of element of $\{0,1\}^*$
with respect of the ordering $\preceq_1$, and $z_1,z_2,...$ be an
enumeration of element of $\{0,1\}^*$ with respect of the ordering
$\preceq_2$. For each $n\in \mathbb{N}$, let $s_n=(\langle
R(y_1),R(y_2),...,R(y_n)\rangle, j)$, and $s'_n=(\langle
R(z_1),R(z_2),...,R(z_n)\rangle, j)$. Let
$Path_1=\{s_1,s_2,...\}$, and $Path_2=\{s'_1,s'_2,...\}$. We are
done.
\end{proof}

Let $(K,s)$ be an arbitrary state of a persistently evolutionary
Kripke model. One may easily observe that for all formula $\psi$,
$(K,s)\models  \Box C_f\varphi\rightarrow \Box\neg K_{ag}\psi$. It
is because if $(K,s)\models   C_f\psi$ then $(K,s)\not\models
\psi$ and thus $(K,s)\not\models K_{ag}\psi$.

Therefore, for the formula $\varphi$ in theorem~\ref{fr}, we have
$(K',s)\models \Box\neg K_{ag} \varphi \wedge \Box\neg K_{ag} \neg
\varphi$. It says that the agent $ag$ never have evidence for
$\varphi$ and never have evidence for $\neg\varphi$. Therefore the
principle of ``from perpetual ignorance to negation" (PIN, see
Chapter~5 of~\cite{kn:ob}) is not true in persistently
evolutionary Kripke models.

\section{ A Persistently Evolutionary Kripke Structure for
 Computation environments}\label{compu}

In this section, we propose a persistently evolutionary Kripke
structure, $K_{ce}$, for the notion of computation environments.
The language of computation environment $L_{ce}$ contains

\begin{itemize}
\item a predicate symbol $SB$ for the successful box,

\item a function symbol $TB$ for the transition box.

\end{itemize}
Let $INST_s$ and $CONF_s$ be two set introduced in the Turing
computation environment (see example~3.4 of~\cite{kn:comp1}). The
set of actions is defined to be \begin{itemize}\item[]
$A_{ce}=\{SB(C)\mid C\in CONF_s\}\cup \{ TB(C,\iota)\mid C\in
CONF_s, \iota\in INST_s\}$.
\end{itemize}
The set of atomic proposition of the Kripke structure $K_{ce}$ is
defined to be
\begin{itemize}
\item[] $P_{ce}=\{SB(C)=b\mid b\in\{0,1\}, C\in
CONF_s\}\cup\{TB(C,\iota)=C'\mid C,C'\in CONF_s, \iota\in
INST_s\}$.
\end{itemize}

The set of meaning functions $\Pi_{ce}$ is defined to be the set
of all functions $\pi$ which satisfy persistently evolutionary
condition, and for every $s\in S_{ce}$,
$C=(q,xb_1\underline{a}b_2y)\in CONF_s$, and $\iota\in INST_s$,

\begin{itemize}
\item[] if $\pi(s,SB(C))=\{ SB(C)=1\}$  then either
$C=(h,\underline{\triangle}x)$ or $C=(h,x\underline{\triangle})$,

\item[] if $C=(h,\underline{\triangle}x)$ then $\pi(s,SB(C))=\{
SB(C)=1\}$,

\item[]
$\pi(s,TB(C,\iota))=\{TB(C,\iota)=(p,xb_1c\underline{b_2}y)\}$ for
$\tau=[(q,a)\rightarrow (p,c,R)]$,

\item[]
$\pi(s,TB(C,\iota))=\{TB(C,\iota)=(p,x\underline{b_1}cb_2y)\}$ for
$\tau=[(q,a)\rightarrow (p,c,L)]$.

\item[] if $\tau\neq[(q,a)\rightarrow (p,c,L)]$ and
$\tau\neq[(q,a)\rightarrow (p,c,R)]$ then
$\pi(s,TB(C,\iota))=\{TB(C,\iota)=\perp\}$.
\end{itemize}


  Let $\pi_i$ be the meaning function that behaves accord to
$SBOX_s$ and $TBOX_s$ of the Turing computation environment. We
prove

\begin{center}$(K_{ce},(\langle\rangle,i))\models \neg K_{ag}
(\mathrm{P=NP})$.\end{center}

  To do this, we should prove that
for every finite sequence of actions $\vec{a}$ in $A_{ce}$, there
exists a meaning function $\pi_j$ such that
$(\vec{a},i)\sim_{ag}(\vec{a},j)$, and
$(K_{ce},(\vec{a},j))\not\models (\mathrm{P=NP})$.

Suppose $\vec{a}=\langle a_1,a_2,...,a_n\rangle$, and let
$H=\{a_i\mid a_i= SB(h,x\underline{\triangle})\}$. We construct a
meaning function $\pi_j$ that considers the following boxes. For
the symbol function $TB$ the meaning function $\pi_j$ behaves
based on the transition box $TBOX_s$. For the symbol predicate
$SB$, it behaves as follows: We persistently evolve the
persistently evolutionary machine $PT_1$ in the way that for every
$x\in \Sigma^*$, if there exists a configuration
   $C=(h,x\underline{\triangle})\}$ such that $SB(C)\in H$, then the
   machine $PT_1$, after evolution, outputs $1$ for $x$ if and
   only if $\pi_i((\vec{a},i), SB(C))=1$~\footnote{Actually,  since $\pi_i$
   is the meaning function that accords with the $SBOX_s$ of the
   Turing computation  environment, for all $C=(h,x\underline{\triangle})\}$ such that $SB(C)\in
   H$, we have $\pi_i((\vec{a},i), SB(C))=1$.}. Then we construct
   a successful box, denoted by  $SBOX'$, which its inner structure is similar  to $SBOX_e$
   except that instead of the $PT_1$ machine, we replaced the above
   evolved $PT_1$ machine. Now, we let $\pi_j$ be the meaning
   function that behaves accord to $SBOX'$. Then the two
   followings are straightforward.

   \begin{itemize}
\item[1-] $(\vec{a},i)\sim_{ag}(\vec{a},j)$, and

\item[2-]$(K_{ce},(\vec{a},j))\not\models (\mathrm{P=NP})$
   \end{itemize} One should verify  that at the state $(\vec{a},j))$, the
   formula $\mathrm{P=NP}$ conflicts with the free will of the
   agent (see the proof of theorem~5.8 of~\cite{kn:comp1}), and thus we have $(K_{ce},(\vec{a},j))\not\models
   (\mathrm{P=NP})$. Therefore,\begin{center}$(K_{ce},(\langle\rangle,i))\models \neg K_{ag}
(\mathrm{P=NP})$.\end{center}
   The finite  sequence $\vec{a}$ was assumed to be arbitrary. Therefore,
   we proved that for all finite sequence of actions $\vec{a}$, $(K_{ce},(\langle\rangle,i))\models \neg K_{ag}
(\mathrm{P=NP})$, and it informally means that
\begin{itemize}
\item[] the agent $ag$ can never know (have evidence for)
$\mathrm{P=NP}$.
\end{itemize}


\begin{thebibliography}{10}

\bibitem{kn:ob} M.~van~Atten,
{\bf On Brouwer}, Wadsworth Philosophers Series, 2004.

\bibitem{kn:BH} M.~van~Atten,
{\bf Brouwer Meets Husserl, on the phenmenology of choice
sequences}, Springer, 2007.

\bibitem{kn:plato} A. Biletzki,  and A. Matar, {\em Ludwig Wittgenstein}, The
Stanford Encyclopedia of Philosophy (Summer 2011 Edition), Edward
N. Zalta (ed.), URL =
http://plato.stanford.edu/archives/sum2011/entries/wittgenstein.


\bibitem{kn:kasir} Ernst Cassirer,
{\bf The Philosphy of Symbolic Forms}.


\bibitem{kn:dit3} H. van Ditmarsch, W. van der Hoek, and B. Kooi,
{\bf Dynamic Epistemic Logic}, Springer, 2008.

\bibitem{kn:model} D. Marker, {\bf Model Theory: an introudction},
Springer, 2002.

\bibitem{kn:comp1} R. Ramezanian, {\em Computation Environments
(1), An interactive semantics for Turing machines (which $P$ is
not equal to $NP$ considering it)}, arXiv: 1205.5994.

\bibitem{kn:TD} A. S. Troelstra, D. van Dalen,
{\bf Constructivism in Mathematics, An introduction}, Vol. 1,
North-Holland, 1988.
\end{thebibliography}
\end{document}